\documentclass{article}

\usepackage{arxiv}

\usepackage[utf8]{inputenc} 
\usepackage[T1]{fontenc}    
\usepackage{hyperref}       
\usepackage{url}            
\usepackage{booktabs}       
\usepackage{amsfonts}       
\usepackage{nicefrac}       
\usepackage{microtype}      
\usepackage{lipsum}		
\usepackage{graphicx}
\usepackage{natbib}
\usepackage{doi}

\usepackage{algorithmic}
\usepackage{mathtools}
\usepackage{amssymb,amsthm, mathtools}
\usepackage{relsize}
\usepackage{subfigure}

\newtheorem*{thm}{Theorem}
\newtheorem*{ex}{Example}

\title{A Non-Binary Method for Finding Interpolants:\\ Theory and Practice}

\author{ \href{https://orcid.org/0000-0003-4170-8665}{\includegraphics[scale=0.06]{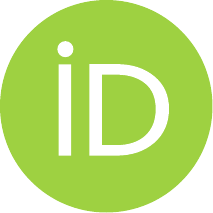}\hspace{1mm}Adam Trybus} \\
	Institute of Philosophy\\
	Jagiellonian University\\
	52 Grodzka St., Kraków 31-044, Poland\\
	\texttt{adam.trybus@uj.edu.pl} \\
	\And
	\hspace{1mm}Karolina Rożko \\
    The University of Zielona Góra, Poland. \\
    \AND
	\hspace{1mm}Tomasz Skura \\
    The University of Zielona Góra, Poland. \\
}

\renewcommand{\shorttitle}{A Non-Binary Method for Finding Interpolants}

\hypersetup{
pdftitle={A Non-Binary Method for Finding Interpolants: Theory and Practice},
pdfsubject={cs.LO},
pdfauthor={Adam Trybus, Karolina Rożko, Tomasz Skura},
pdfkeywords={proof systems, refutation systems, interpolation, classical propositional logic},
}

\begin{document}
\maketitle

\begin{abstract}
	We describe a new method of finding interpolants for classical logic using certain refutation system as a starting point. Refutation can be thought of as an alternative approach to the analysis of formal systems: instead of focusing on which formulas provably belong to a given logic, it shows which formulas are to be rejected. Thus, it provides a mirror proof system. As it turns out, the benefits of such an approach go well beyond the area of refutation calculi themselves. We provide one such example in the shape of an interpolant-searching method. To be sure, a number of such methods are already in use. The novelty of our proposal lies in the fact that it can be considered as based on a non-binary version of resolution.
\end{abstract}


\section{Introduction}
Our goal is twofold. First, to describe a new method of finding interpolants in classical propositional logic; and second, to demonstrate the practical applicability of this method by providing an appropriate implementation and testing its performance.

In broad strokes, the interpolation theorem for FOL states that given two formulas $A$ and $B$ sharing at least one predicate symbol and such that $A \to B$ is valid, there exists a formula $C$ (called \emph{an interpolant}) built from the predicate symbols shared by $A$ and $B$, such that $A \to C$ and $C \to B$ are valid (see \cite{Craig1957a}). In case of classical propositional logic, the condition is that $A$ and $B$ have to share at least one propositional variable.

The main theoretic result of this paper is a new proof of the theorem. Our approach stems from the research on the so-called \emph{refutation systems}, where the focus is on rules of \emph{rejection} of formulas. A refutation system is a collection of non-valid formulas (refutation axioms) together with (refutation) rules preserving non-validity. This concept was introduced by \L ukasiewicz in \cite{Lukasiewicz1957} (for more on refutation systems see e.g. \cite{Skura2011}). We make use of such a system designed for classical propositional logic (see \cite{Skura2013}, p. 110 -- 115) and show how it can be used in an interpolant-searching procedure. We implement this procedure in a form of a Python script and present the results of an analysis of its performance.

\section{Preliminaries}

We define the syntactical and semantic notions in line with the usual practice.

\paragraph{Syntax.} Lowercase letters $p, q, \ldots$ --- called \emph{propositional variables} --- together with two \emph{constants} $\bot$ (\emph{Falsum}) and $\top$ (\emph{Verum}) form a set of \emph{atomic formulas} (\textit{AT}). The set \textit{FOR} of \emph{well-formed formulas} (the elements of which are denoted by uppercase letters $A, B, \ldots$ possibly with subscripts) is generated in a standard way from the set \textit{AT} and the following connectives: $\neg$ (negation); $\wedge$ (conjunction); $\vee$ (disjunction); $\rightarrow$ (implication) and $\equiv$ (equivalence). A \emph{literal} $l$ is either a propositional variable (a \emph{positive} literal) or its negation (a \emph{negative} literal). We call \textit{l}$^*$ the \emph{complement} of $l$ and define it as follows. If $l$ is a positive literal $p$, then $l^*$ is $\lnot p$ and if $l$ is a negative literal $\lnot p$, then $l^*$ is $p$. In the remainder of the article, we assume that whenever $\mathrm{m},\mathrm{n}$ are used, these stand for natural numbers: $\mathrm{m}, \mathrm{n} \geq 0$.

We denote sets (most cases, sets of formulas) in the following way:\linebreak $\mathcal{X}, \mathcal{Y}, \ldots$, if needs be extending this notation with subscripts. Let $\mathcal{X}=\{A_1, ...,A_{\mathrm{n}} \}$, where each $A_i \in FOR$. The expression $\bigwedge \mathcal{X}$ stands for the formula $A_1 \wedge~...~\wedge A_{\mathrm{n}}$, whereas $\bigvee \mathcal{X}$ stands for $ A_1 \vee ... \vee A_{\mathrm{n}}$. If $\mathcal{X}$ is a finite (possibly empty) set of literals, then by $\mathcal{X}^*$ we denote the set where every literal of $\mathcal{X}$ is replaced by its complement and call $\bigwedge \mathcal{X}^*$ a \emph{complement} of $\bigvee \mathcal{X}$ and $\bigvee \mathcal{X}^*$ a \emph{complement} of $\bigwedge \mathcal{X}$. (If $\mathcal{X} = \emptyset$, then $\bigvee \mathcal{X} = \bot$ with $\bigwedge \mathcal{X}^* = \top$ and $\bigwedge \mathcal{X} = \top$ with $\bigvee \mathcal{X}^* = \top$.) We extend this notation in the following way: when $\mathcal{X}$ is a set of formulas $A_1, \ldots, A_{\mathrm{n}}$ that are all either disjunctions or conjunctions of literals, then by $\mathcal{X}^*$ we mean the set containing $A_1^*, \ldots, A_{\mathrm{n}}^*$. 

A \emph{clause} is a formula of the form $\bigvee \mathcal{S}$, where $\mathcal{S}$ is a finite set of literals.\footnote{Note that our definitions do not allow any given literal to occur in a clause more than once.} We say that a formula $A$ is in a \emph{conjunctive normal form} (CNF), if $A$ is of the form $B_1 \wedge \ldots \wedge B_{\mathrm{m}}$, where each $B_{\mathrm{i}}$ is a clause; $A$ is in a \emph{disjunctive normal form} (DNF), if $A$ is of the form $B_1 \vee \ldots \vee B_{\mathrm{m}}$, where each $B_{\mathrm{i}}$ is the complement of a clause.

We introduce the following notational convention. Let $\mathcal{X} = \{A_1, \ldots, A_{\mathrm{n}} \}$ and $\mathcal{Y} = \{B_1, \ldots, B_{\mathrm{m}} \}$. We define $\mathcal{X} \longrightarrow \mathcal{Y}$ to be $\bigwedge \mathcal{X} \to \bigvee \mathcal{Y}$. We also write $A_1, \ldots, A_{\mathrm{n}} \longrightarrow B_1, \ldots, B_{\mathrm{m}}$, for $\mathcal{X} \longrightarrow \mathcal{Y}$. Any formula of this shape with $\bigwedge\mathcal{X}$ in CNF and $\bigvee \mathcal{Y}$ in DNF is said to be simply in a \emph{normal form}. Finally, consider a formula $\bigwedge \mathcal{X} \rightarrow \bigvee \mathcal{Y}$ in a normal form and put $\mathcal{Z} = \mathcal{X} \cup \mathcal{Y}^*$. Any formula of the shape $\mathcal{Z} \longrightarrow \bot$ is said to be in a \emph{clausal normal form}. The \emph{rank} of a (formula in a) clausal normal form is the number of literals $l$ such that $l$~and~$l^*$ occur in distinct clauses in $\mathcal{Z}$.


\paragraph{Semantics.} Consider the set $\mathcal{V}$ of all \emph{classical valuation} functions $v:  AT \longrightarrow \{1,0\}$ (that is the functions extended to the set $FOR$ by means of classical interpretation of the connectives). We say that $A$ is \emph{valid} and write $\models A$, if $v(A) = 1$ for any classical valuation function $v \in \mathcal{V}$. The set of all valid formulas will be denoted as  $\mathcal{CL}$. We recall that every formula $A$ can be converted to a formula $A'$ with the property $\models A \equiv A'$ and such that $A'$ is either in CNF or DNF.\\

As usual, we abbreviate the expression `if and only if' as 'iff'.

\section{Previous work}

There is no shortage of interpolation methods. The more well-known theoretical approaches include those described by Maehara (see \cite{Takeuti1987}), Kleene (see \cite{Kleene1967}) and Smullyan (see \cite{Smullyan1968}). In this article, however, we are more interested in the practically-applicable ones. The article \cite{Bonacina2015} is a recent survey of the most popular approaches in terms of what the authors refer to as the \emph{interpolation systems} (essentially: procedures for finding interpolants) as used in the context of automated theorem-proving. There, one finds a description of a system that works by producing interpolants of the same shape as is the case with our method. On that level, the differences between our approach and the one presented there are only superficial, essentially boiling down to notation. This system is also described in \cite{DSilva2010}, where it is referred to as HKP-system\footnote{The acronym comes from the initials of the researchers who independently worked on this system: Huang (\cite{Huang1995}), Kraji\v{c}ek (\cite{Krajicek1997}) and Pudl\'{a}k (\cite{Pudlak1997}).} and is presented as one of the two main approaches. The first important difference, however, is that the proofs of the interpolation theorem that form the basis of this system are arguably quite involved (part of the reason being that these are presented in a more general context of first-order logic) and thus do not seem to be geared towards potential practical uses. More importantly though, the system works by utilising resolution proofs in creating interpolants. And this we consider a crucial difference between what we propose and HKP-system (this applies in fact to any system described in \cite{DSilva2010}). In our approach, starting from more or less standard refutation systems for classical logic we arrived at an interpolation system that can be thought of as utilising \emph{non-binary} resolution. Therefore, despite a number of similarities, our proposal stands out in the following ways: (i) the proof of the main theorem is simple and (ii) lends itself well to potential implementations due to its constructive manner, and (iii) the fact that our system does not use binary resolution means that it has the potential to produce results in a smaller number of steps (we come back to this point later).

\section{Refutation}

We take for granted all the notions used in the standard approach to proof theory and apply them in the setting of refutation. Intuitively, the idea is to propose a sort of a mirror system, where instead of proving that certain formulas are valid, we show that they can be refuted, i.e. are non-valid. By analogy to a more standard proof system, a \textit{refutation system} is a pair that consists of~a set of~\textit{refutation axioms} and a set of~\textit{refutation rules}. The rules are of the following shape: 

\begin{center}
\noindent $\mathlarger{\frac{\textstyle B_1;\ldots;B_{\mathrm{m}}}{\textstyle B}}$,\\
\end{center}

\noindent where all the elements are formulas. We say that a formula $A$ is \emph{refutable}, denoted $\dashv A$, if $A$ is derivable in this system. In the context of classical propositional logic we are interested in refutation systems with the following property: for any formula $A$ we have $\dashv A$ iff $A \notin \mathcal{CL}$. To this end, we need to ensure that the refutation axioms themselves are non-valid and that the refutation rules preserve non-validity. Consider the following refutation system introduced in \cite{Skura2013} (p. 111).\\


{\textit{Refutation axioms:}} every clausal normal form $\mathcal{Z} \longrightarrow \bot$ of rank $0$ such that $\bot \notin \mathcal{Z}$.  

\bigskip

{\textit{Refutation rules:}}\\
\begin{center}

\noindent $\mathlarger{\frac{C_1,...,C_{\mathrm{m}},\mathcal{E} \longrightarrow \bot}{l \vee C_1,~...,~l \vee C_{\mathrm{m}}, l^* \vee D_1,~...,~l^* \vee D_{\mathrm{n}}, \mathcal{E} \longrightarrow \bot}}$\\

\bigskip

\noindent $\mathlarger{\frac{D_1,...,D_{\mathrm{n}},\mathcal{E} \longrightarrow \bot}{l \vee C_1,~...,~l \vee C_{\mathrm{m}}, l^* \vee D_1,~...,~l^* \vee D_{\mathrm{n}}, \mathcal{E} \longrightarrow \bot}}$,\\

\end{center}

\bigskip

\noindent where all $C_{\mathrm{i}}$, $D_{\mathrm{j}}$ are clauses and $\bigwedge \mathcal{E}$ is a CNF such that neither $l$ nor $l^*$ occurs in the elements of $\mathcal{E}$. In \cite{Skura2013} (p. 112) it is proved that $A \notin \mathcal{CL}$ iff $\dashv A$. A note on how the system operates. A quick glance ensures us that these rules share the bottom element and differ only in the top. Essentially, the rules work by eliminating a pair of literals $l, l^*$ from the bottom formula. We can think of them as having the following shape: $\frac{\textstyle F_1}{\textstyle F}$ and $\frac{\textstyle F_2}{\textstyle F}$, respectively. Now, $F_1$, apart from the neutral $\mathcal{E}$, contains what is left from the clauses that contained $l$ in $F$, whereas $F_2$ on top of $\mathcal{E}$ consists of what is left from the clauses that contained $l^*$ in $F$. We have that $F$ is not valid iff either $F_1$ is not valid or $F_2$ is not valid. This leads us to the following property of this refutation system:\\

($\dagger$) $F$ is valid iff both $F_1$ and $F_2$ are valid.\\

As we shall see, the property ($\dagger$) is crucial for our purposes. Under a different guise, it has been explored in \cite{Davis1959} but to the best of our knowledge it has not been used in any interpolant-finding procedure so far. The refutation axioms used in this system can, admittedly, be quite long but their advantage from our point of view is that these can be easily seen to be non-valid, since in such a formulas, there is no pair of literals $l, l^*$ occurring in distinct clauses.

\section{Theoretical results --- the interpolant-finding procedure}
In order to use the above refutation system as a starting point of searching for interpolants, it is perhaps more convenient to re-write it in the following way.\\
\bigskip

{\textit{Refutation axioms:}} every normal form $\mathcal{X} \longrightarrow \mathcal{Y}$ of rank $0$, such that $\bot \notin \mathcal{X}$ and $\top \notin \mathcal{Y}$. \\ 
\bigskip

{\textit{Refutation rules:}}\\

\begin{center}
\begin{scriptsize}

\noindent $\mathlarger{\frac{A_1,\ldots , A_{\mathrm{m}}, \mathcal{E}_1 \longrightarrow C_{1},\ldots , C_{\mathrm{p}}, \mathcal{E}_2}{A_1 \vee l \ldots A_{\mathrm{m}} \vee l, B_1 \vee l^*,\ldots , B_{\mathrm{o}} \vee l^*, \mathcal{E}_1 \longrightarrow C_{1} \wedge l^*,\ldots , C_{\mathrm{p}} \wedge l^*, D_{1} \wedge l, \ldots ,D_{\mathrm{r}} \wedge l, \mathcal{E}_2}}$\\
\bigskip


\noindent $\mathlarger{\frac{B_1,\ldots , B_{\mathrm{o}}, \mathcal{E}_1 \longrightarrow D_{1},\ldots ,D_{\mathrm{r}}, \mathcal{E}_2}{A_1 \vee l \ldots A_{\mathrm{m}} \vee l, B_1 \vee l^*,\ldots , B_{\mathrm{o}} \vee l^*, \mathcal{E}_1 \longrightarrow C_{1} \wedge l^*,\ldots , C_{\mathrm{p}} \wedge l^*, D_{1} \wedge l,\ldots ,D_{\mathrm{r}} \wedge l, \mathcal{E}_2}}$,\\
\bigskip

\end{scriptsize}
\end{center}

\noindent where $A_{\mathrm{i}}$, $B_{\mathrm{j}}$ are clauses, $C_{\mathrm{k}}$, $D_{\mathrm{l}}$ are complements of clauses, $\mathcal{E}_1$ is a set of clauses, $\mathcal{E}_2$ is a set of complements of clauses and neither $l$ nor $l^*$ occurs in $\mathcal{E}_1$ and $\mathcal{E}_2$. Intuitively, these rules can be described as follows. Each application of either rule works by eliminating $l, l^*$ from the bottom formula, which is of the form $\mathcal{X} \longrightarrow \mathcal{Y}$. Observe that the rules differ only in their top elements. In both cases $\mathcal{E}_1$ and $\mathcal{E}_2$ are present but apart from that the first rule contains only the remainders of the clauses from $\mathcal{X}$ that contained $l$ and the remainders of the clauses from $\mathcal{Y}$ that contained $l^*$, whereas the second rule contains only the remainders of the clauses from $\mathcal{X}$ that contained $l^*$ and the remainders of the clauses from $\mathcal{Y}$ that contained $l$. We can represent the above rules as $\frac{\textstyle G_1}{\textstyle G}$ and $\frac{\textstyle G_2}{\textstyle G}$, respectively.\\

\noindent Now, we get the following version of the ($\dagger$) property:\\

($\ddag$) $G$ is valid iff both $G_1$ and $G_2$ are valid.\\

Note that for our purposes this version of the system is primary to the original formulation we have presented above. In fact, the rules $\frac{\textstyle F_1}{\textstyle F}$ and $\frac{\textstyle F_2}{\textstyle F}$ can be viewed as resulting from all the formulas in $\frac{\textstyle G_1}{\textstyle G}$ and $\frac{\textstyle G_2}{\textstyle G}$ being changed from normal form to clausal normal form. Since our aim is to describe an interpolant-finding procedure, we --- naturally --- focus on normal forms. However, as it will become evident, we also make heavy use of clausal normal forms. 

Consider formulas $X$ and $Y$, such that $\models X \to Y$. According to what has been stated above, we can assume that $X$ is in CNF and $Y$ is in DNF. Then we see that this formula can in fact be represented as $\mathcal{X} \longrightarrow \mathcal{Y}$. If $X$ and $Y$ have no propositional variables in common, then we have that either (a) $\models \bigvee \mathcal{Y}$ or (b) $\models \neg \bigwedge \mathcal{X}$. In case (a) $\models \mathcal{X} \longrightarrow \top$ and $\models \top \longrightarrow \mathcal{Y}$, whereas in case (b) $\models \mathcal{X} \longrightarrow \bot$ and $\models \bot \longrightarrow \mathcal{Y}$. 

Seeing that when dealing with the formulas in normal forms, one can still obtain interpolant-like formulas, we extend the definition of an interpolant by not insisting that the two formulas share variables. We say that $I$ is an \emph{(extended) interpolant} of $X$ and $Y$, if both $\models X \to I$ and $\models I \to Y$ and if there is a pair of literals $l, l*$ occurring in different clauses of $\mathcal{X}\cup\mathcal{Y}^*$, then the variables of $I$ are among those common to $X$ and $Y$.\\

We have the following theorem.


\begin{thm}[Extended Interpolation Theorem]
All formulas $X$ and $Y$ such that $\models X \to Y$ have an interpolant.
\end{thm}

\begin{proof} Let $G = \mathcal{X} \longrightarrow \mathcal{Y}$ be a formula in a normal form constructed on the basis of $X$ and $Y$ (as indicated above) such that $\models (X \to Y) \equiv G$ and let $G' = \mathcal{X} \cup \mathcal{Y}^* \longrightarrow \bot$ be a respective clausal normal form.\footnote{It is very easy to change from a normal form to a clausal normal form: this is done simply by complementing $\mathcal{Y}$ and `shifting' it to the other side of the implication. In the proof, we freely move from one form to another whichever is more convenient for the particular job at hand.} Note that obviously $\models G$ and $\models G'$. The proof proceeds by induction on the rank $\mathrm{n}$ of $G'$.\\

\noindent (i) $\mathrm{n} = 0$.\\

\noindent Recall that since $\models G$, we get that either $\models \neg \bigwedge \mathcal{X}$ or $\models \bigvee \mathcal{Y}$ . Thus, the empty clause has to be either in $\mathcal{X}$ or in $\mathcal{Y}^*$:\footnote{Note that this is irrespective of whether $\mathcal{X}$ and $\mathcal{Y}$ have any variables in common. However, the implementation has to take a bit more nuanced approach where all the subcases are dealt with (essentially amounting to the same, see the description of the algorithm below).}\\

(Case 1) $\bot \in \mathcal{X}$. Then $\models \mathcal{X} \longrightarrow \bot$ and $\models \bot \longrightarrow \mathcal{Y}$. Thus the interpolant is $\bot$.\\

(Case 2) $\bot \in \mathcal{Y}^*$. Then $\models \mathcal{X} \longrightarrow \top$ and $\models \top \longrightarrow \mathcal{Y}$. Thus the interpolant is $\top$.\\

\noindent (ii) $\mathrm{n} > 0$.\\ 

\noindent Assume that the theorem holds for clausal normal forms of rank $< \mathrm{n}$.\\

\noindent Since $\mathrm{n} > 0$, there exists a literal $l$ such that $l$ and $l^*$ occur in distinct clauses of $\mathcal{X} \cup \mathcal{Y}^*$. Therefore, $G$ has the following form.

\begin{scriptsize}

$$G = A_1 \vee l~\ldots~A_{\mathrm{m}} \vee l, B_1 \vee l^*,..., B_{\mathrm{o}} \vee l^*, \mathcal{E}_1 \longrightarrow C_{1} \wedge l^*, \ldots, C_{\mathrm{p}} \wedge l^*, D_{1} \wedge l, \ldots,D_{\mathrm{r}} \wedge l, \mathcal{E}_2$$

\end{scriptsize}

(In short, this can be represented as $\mathcal{A},\mathcal{B},\mathcal{E}_1 \longrightarrow \mathcal{C},\mathcal{D},\mathcal{E}_2$.)\\


\noindent Thus, we can apply the refutation rules described above, obtaining:

$$G_1 = A_1,..., A_{\mathrm{m}}, \mathcal{E}_1 \longrightarrow C_{1}, ..., C_{\mathrm{p}}, \mathcal{E}_2$$

and

$$G_2 = B_1,..., B_{\mathrm{o}}, \mathcal{E}_1 \longrightarrow D_{1}, ...,D_{\mathrm{r}}, \mathcal{E}_2.$$


Note that in the above we do not exclude a situation when $l$ or $l^*$ are absent from either $\mathcal{X}$ or $\mathcal{Y}^*$  --- in effect making any of the $\mathcal{A},\mathcal{B}, \mathcal{C}$ or $\mathcal{D}$ empty --- as long as one finds $l$ and $l^*$ in different clauses.

Since $\models G$, by ($\ddag$) both $\models G_1$ and $\models G_2$. Now $G_1$ and $G_2$ can easily be transformed to clausal normal forms of rank~$< \mathrm{n}$ hence, by the inductive hypothesis, there are formulas $I(G_1)$ and $I(G_2)$, such that:\\

\noindent (I) $\models A_1, ..., A_{\mathrm{m}}, \mathcal{E}_1 \to I(G_1)$ and $\models I(G_1) \to C_{1}, \ldots, C_{\mathrm{p}}, \mathcal{E}_{2}.$\\

(In short: $\models \bigwedge \mathcal{X}_1 \to I(G_1)$ and $\models I(G_1) \to \bigvee \mathcal{Y}_1$.)\\

\noindent (II) $\models B_1,..., B_{\mathrm{o}}, \mathcal{E}_1 \to I(G_2)$ and $\models I(G_2) \to D_{1}, \ldots,D_{\mathrm{r}}, \mathcal{E}_{2}.$\\

(In short: $\models \bigwedge \mathcal{X}_2 \to I(G_2)$ and $\models I(G_2) \to \bigvee \mathcal{Y}_2$.)\\

\noindent Now, define $I(G)$ in the following manner.

\begin{equation*}
  I(G) =
  \begin{cases}
    I(G_1)\vee I(G_2) & \text{if neither $l$ nor $l^*$ occurs in $\mathcal{Y}^*$} \\
    I(G_1)\wedge I(G_2) & \text{if neither $l$ nor $l^*$ occurs in $\mathcal{X}$} \\
    (l \vee I(G_1)) \wedge (l^* \vee I(G_2)) & \text{otherwise}\\
  \end{cases}
\end{equation*}

\noindent We claim that $I(G)$ is an interpolant for $X$ and $Y$. Observe that since by definition $I(G_1)$ is built up from the variables common to $\mathcal{X}_1$ and $\mathcal{Y}_1$\footnote{We take liberties in applying this notion. We obviously mean that the variables occur in the elements of the specified sets but saying so a number of times adds unnecessary clutter.} and $I(G_2)$ is built up from the variables common to $\mathcal{X}_2$ and $\mathcal{Y}_2$, the resulting $I(G)$ is built up from the variables common to $\mathcal{X}$ and $\mathcal{Y}$.

We first deal with the situation, where at least one of $l, l^*$ occurs in $\mathcal{X}$ and similarly at least one of $l, l^*$ occurs in $\mathcal{Y}^*$, giving $I(G) = (l \vee I(G_1)) \wedge (l^* \vee I(G_2))$. In what follows, let $\mathcal{G}_1 = \{A_1 \vee l~\ldots~A_{\mathrm{m}} \vee l, B_1 \vee l^*,\ldots, B_{\mathrm{o}} \vee l^*\}$ and $\mathcal{G}_2 = \{C_{1} \wedge l^*, \ldots, C_{\mathrm{p}} \wedge l^*, D_{1} \wedge l, \ldots,D_{\mathrm{r}} \wedge l\}$ (note that $\mathcal{X} = \mathcal{G}_1 \cup \mathcal{E}_1$ and $\mathcal{Y} = \mathcal{G}_{2} \cup \mathcal{E}_{2}$). Assuming that $I(G)$ is not an interpolant, we have to consider two cases.\\

(Case 1) $\mathcal{G}_1, \mathcal{E}_1 \longrightarrow I(G) \notin \mathcal{CL}$.\\

\noindent In such a case, there has to exist a valuation $v~{\in}~\mathcal{V}$ such that $v(\bigwedge \mathcal{G}_1) = v(\bigwedge \mathcal{E}_1) = 1$ and~$v(I(G)) = 0$. We show that whatever value is assigned by $v$ to $l$, one obtains a contradiction. Assume $v(l) = 1$. This means that $v(l^*) = 0$. Since $v(B_i \vee l^*) =1$ and $v(l^*) = 0$, we have $v(B_i) = 1$ for every $i$. Hence $v(I(G_2)) = 1$ by II.
(If there are no $B_i$, then $I(G_2))$ is true because so is $\bigwedge\mathcal{E}_1$.).\footnote{This, and similar situations are encountered when neither $l$ nor $l^*$ occurs in $\mathcal{X}$ or $\mathcal{Y}^*$.} Thus $v(I(G)) = 1$. This is a contradiction. Assume now that $v(l) = 0$. This means that $v(l^*)=1$. Then since $v(A_i\vee l) = 1$, we have $v(A_i) = 1$ for every $i$. Hence $v(I(G_1)) = 1$ by I. (If there are no $A_i$, then $I(G_1))$ is true because so is $\bigwedge\mathcal{E}_1$.) Thus, $v(I(G)) = 1$. This is also a contradiction.

\bigskip

(Case 2) $ I(G) \longrightarrow \mathcal{G}_{2}, \mathcal{E}_{2} \notin \mathcal{CL}$.\\

\noindent In such a case, there has to exist a valuation $v~{\in~}\mathcal{V}$ such that $v(I(G)) = 1$. This means that $v(l \vee I(G_1)) = v(l^* \vee I(G_2)) = 1$. We also have~$v(\bigvee \mathcal{G}_{2})~= v(\bigvee \mathcal{E}_{2})~=~0$. Again, we show that whatever value is assigned by $v$ to $l$, one obtains a contradiction. Assume $v(l) = 1$. This means $v(l^*) = 0$, immediately giving $v(I(G_2)) = 1$. Thus we have either $v(\bigvee \mathcal{E}_{2}) = 1$ or $v(D_{i}) = 1$ for some~$i$, by~II. But we already have that $v(\bigvee \mathcal{E}_{2}) = 0$, forcing $v(D_{i}) = 1$ for~some~$i$. Hence, for some $i$, $v(D_{i} \wedge l) = 1$. Therefore $v(\bigvee \mathcal{G}_{2}) = 1$. (If there are no $D_i$, then $\bigvee\mathcal{E}_2$ is true because so is $I(G_2)$.) This is a contradiction. Now assume $v(l) = 0$. Then $v(l \vee I(G_1)) = 1 $. Since $v(I(G_1)) = 1$, we have $v(\bigvee \mathcal{E}_{2}) = 1$ or $v(C_{i}) = 1$ for some $i$ by I. But we~already have that $v(\bigvee \mathcal{E}_{2}) = 0$, so $v(C_{i}) = 1$ for some $i$. Hence $v(C_{i} \wedge l^*) = 1$. So $v(\bigvee \mathcal{G}_{2}) = 1$. (If there are no $C_i$, then $\bigvee\mathcal{E}_2$ is true because so is $I(G_1)$.) This is also a contradiction.\\

Now, consider the situation, where neither $l$ nor $l^*$ occurs in $\mathcal{X}$, giving $I(G) = I(G_1) \wedge I(G_2)$. In (Case 1), by the fact that there are neither $A_i$ nor $B_i$ we immediately obtain $v(I(G_1)) = v(I(G_2)) = 1$, by I and II, respectively. This yields $v(I(G)) = 1$, a contradiction. In (Case 2) we immediately get a contradiction by observing that $v(I(G_1)) = v(I(G_2)) = 1$.

Finally, consider the situation, where neither $l$ nor $l^*$ occurs in $\mathcal{Y}^*$, giving $I(G) = I(G_1)\vee I(G_2)$. In (Case 1) the contradiction follows in the same way as in the main line of reasoning. In (Case 2) by the fact that there are neither $C_i$ nor $D_i$ we immediately obtain $v(I(G_1)) = v(I(G_2)) = 1$, by I and II, respectively. This also yields $v(I(G)) = 1$, a contradiction. \end{proof}

It should be obvious how the procedure operates given a formula satisfying the conditions of the theorem. If the rank of this formula is $0$, we obtain $\bot$ or $\top$ as the interpolant, otherwise we apply the ($\ddag$) property to eliminate a pair of literals (shifting from normal to clausal normal forms in the process). We continue until we obtain a number of formulas of rank equal $0$ and on this basis built the interpolant for the initial formula. Let us note that the procedure can work quite rapidly, given that certain relations between the formulas on both sides of the implication are taken into account. Consider the following example.\\

\begin{ex}
{\rm Suppose that we want to find an interpolant for the formulas $\bigwedge \mathcal{X}$ and $\bigvee \mathcal{Y}$, where: 
$$\mathcal{X} = \{p \vee \neg s, q \vee r, \neg p \vee \neg q, \neg r \vee s \}$$ and $$\mathcal{Y} = \{ \neg p \wedge r, \neg p \wedge q, p \wedge \neg q \}.$$

(The formula $G = \mathcal{X} \longrightarrow \mathcal{Y}$ is a clausal normal form of rank $4$). Say we choose $p$ and $\lnot p$ as the pair of literals to be eliminated. We obtain the following:
\begin{enumerate} 
\item $G_1$ is $ \neg s, q \vee r, \neg r \vee s  \longrightarrow r, q $. Seeing that the formula $r \vee q$ occurs both in $X_1$ and $Y_1$, we choose this formula as an interpolant for $X_1$ and $Y_1$. 
\item $G_2$ is $q \vee r, \neg q, \neg r \vee s \longrightarrow \neg q$. Seeing that the formula $\neg q$ occurs both in $X_2$ and $Y_2$, we choose this formula as an interpolant for $X_2$ and $Y_2$. 
\item The interpolant for $X$ and $Y$ has the following shape: $(p \vee r \vee q) \wedge (\neg p \vee \neg q)$.\\

\end{enumerate}
}
\end{ex}

Thus, we were able to skip some parts of the procedure using convenient shortcuts and significantly reduce the number of steps required. In the next section we show how the procedure works without taking advantage of such solutions.

\section{Implementation}

We have implemented the above procedure in the form of a Python script. Before we move on to the presentation of the implementation itself, a few words of caution are in order. As mentioned, the above example contains certain simplifications, justified from our human perspective as helping in shortening the search for an interpolant, but which nevertheless means that the original procedure is altered. The task of presenting the procedure in full, not optimised for human readability, is however, crucial for building its practical implementation. It is important, as it helps one to appreciate and understand the inner workings of the original procedure, providing a clear picture of what is happening. Another reason for doing so is that some of the perceived simplifications might turn out not to be that beneficial from the point of view of how the implementation works. Let us now revisit the example above to clarify the matter. This will also be an occasion for us to describe the notation used in our program. Recall that in this case we were to find an interpolant for (assuming the notational conventions carry over): $\mathcal{X} = \{p \vee \neg s, q \vee r, \neg p \vee \neg q, \neg r \vee s \}$ and $\mathcal{Y} = \{ \neg p \wedge r, \neg p \wedge q, p \wedge \neg q \}$. This in our notation becomes \texttt{[D.p..Ns., D.q..r., D.Np..Nq., D.Nr..s.]} and \texttt{[C.p..Nr.,C.Np..q., C.p..Nq.]} respectively. Now, the interpolant procedure as described above is quite short, due to the fact that at some point we find the same formulas on both sides of the implication. However, the pure procedure does not make use of this convenient shortcut. In fact, the full interpolant searching procedure looks as follows. As above, let us choose \texttt{.p.} and \texttt{.Np.} as the pair of literals to be eliminated. We represent this as follows:\\

\begin{scriptsize}


\begin{center}

\texttt{X:  [D.p..Ns., D.q..r., D.Np..Nq., D.Nr..s.] Y:  [C.p..Nr., C.Np..q., C.p..Nq.]}\\
\bigskip
\end{center}
\begin{tiny}
\texttt{X:  [D.Ns., D.q..r., D.Nr..s.] Y:  [C.q.]} \hfill \texttt{X:  [D.Nq., D.q..r., D.Nr..s.] Y:  [C.Nr., C.Nq.]}\\
\end{tiny}

\end{scriptsize}

Note that, in each computing step we produce \texttt{Y'} according to the rules set out above --- in short by replacing all \texttt{C}'s with \texttt{D}'s and reversing negations (shifting the contents of \texttt{Y} to the left). Now, where humans might take a shortcut, the full procedure must continue as long as there are literals to be eliminated. Therefore we execute a recursive call on the resulting two formulas without \texttt{.p.} and \texttt{.Np.} (note that we eliminate from \texttt{Y'}, therefore what was \texttt{.p.} in \texttt{Y'} becomes \texttt{.Np.} in \texttt{Y} and the other way round). Going in the depth-first fashion (as recursion tends to operate), say that we chose to eliminate the pair \texttt{.q.} and \texttt{.Nq.} from \texttt{X:  [D.Ns., D.q..r., D.Nr..s.] Y:  [C.q.]} obtaining:\\

\begin{scriptsize}


\begin{center}

\texttt{X:  [D.Ns., D.q..r., D.Nr..s.] Y:  [C.q.]}\\
\bigskip
\end{center}

\texttt{X:  [D.r., D.Ns., D.Nr..s.] Y:  []} \hfill \texttt{X:  [D.Ns., D.Nr..s.] Y:  [.1.]}\\

\end{scriptsize}

Consider now the formula \texttt{X:  [D.r., D.Ns., D.Nr..s.] Y:  []}. Since \texttt{Y} represents the empty set, we obviously have that there are no common variables shared by \texttt{X} and \texttt{Y}, thus landing us in one of the base cases for this recursive procedure. The theoretical result guarantees that in any given base case either the negation of \texttt{X} or \texttt{Y} itself is a tautology. Since \texttt{Y} is empty,\footnote{One can think of \texttt{Y} as containing $\bot$. This is in fact the approach taken up in the theoretical part. We decided to not to implement it this way, as requiring an additional (and unnecessary) computational step.} the former has to obtain. Indeed, a glance at \texttt{X} is enough to convince ourselves that its negation has to be tautologous: therefore the interpolant is $\bot$ (or \texttt{.0.} in our notation). Consider now \texttt{X:  [D.Ns., D.Nr..s.] Y:  [.1.]}. In this case we see that \texttt{Y}'s only element is \texttt{.1.} (that is $\top$, it appeared here as a result of the last disjunct of \texttt{Y'} being eliminated in the previous step and the result --- with flipped value --- carrying over to the other side of implication).\footnote{The convention is that empty conjuncts (\texttt{C} without any arguments) equal $\top$ (or \texttt{.1.} in our notation) and empty disjuncts (\texttt{D}) equal $\bot$ (\texttt{.0.}).} Thus we obtain the interpolant in the form of \texttt{.1.}. Going back up, let us consider the rightmost result \texttt{X:  [D.Nq., D.q..r., D.Nr..s.] Y:  [C.Nr., C.Nq.]} and eliminate \texttt{.r.}/\texttt{.Nr.}:\\

\begin{scriptsize}


\begin{center}

\texttt{X:  [D.Nq., D.q..r., D.Nr..s.] Y:  [C.Nr., C.Nq.]}\\
\bigskip
\end{center}

\texttt{X:  [D.q., D.Nq.] Y:  [.1., C.Nq.]} \hfill \texttt{X:  [D.s., D.Nq.] Y:  [C.Nq.]}\\

\end{scriptsize}

Moving forward, note that humans can easily spot that the negation of \texttt{X} in the leftmost formula is a tautology. The procedure, however, has no choice: it eliminates \texttt{.q.}  and \texttt{.Nq.} from \texttt{X:  [D.q., D.Nq.] Y:  [.1., C.Nq.]} without taking heed of this additional piece of information:\\

\begin{scriptsize}


\begin{center}

\texttt{X:  [D.q., D.Nq.] Y:  [.1., C.Nq.]}\\
\bigskip
\end{center}

\texttt{X:  [.0.] Y:  [.1., .1.]} \hfill \texttt{X:  [.0.] Y:  [.1.]}\\

\end{scriptsize}

Only now can we finish, seeing that we have again reached two base cases. With both \texttt{X} and \texttt{Y} not empty the procedure is to check whether the negation of \texttt{X} is a tautology. Since this is the case, we obtain the interpolant to be \texttt{.0.};\footnote{Otherwise we would have that \texttt{Y} is a tautology as guaranteed by the theoretical considerations.} the rightmost formula yields identical result. There is only one formula now to be taken apart and it is \texttt{X:  [D.s., D.Nq.] Y:  [C.Nq.]} and we do so by removing the pair \texttt{.q.} and \texttt{.Nq.}:\\

\begin{scriptsize}


\begin{center}

\texttt{X:  [D.s., D.Nq.] Y:  [C.Nq.]}\\
\bigskip
\end{center}

\texttt{X:  [D.s.] Y:  [.1.]} \hfill \texttt{X:  [.0., D.s.] Y:  []}\\

\end{scriptsize}


Remembering what has been said already, we see that in the leftmost formula \texttt{Y} being a tautology means that the interpolant is \texttt{.1.}, whereas for the rightmost formula we see that the negation of \texttt{X} is a tautology, thus the interpolant is \texttt{.0.}. Finally, we arrive at the interpolant of the following shape (presented in a more standard logical notation):

\begin{footnotesize}

$$\{p \vee [(q \vee \bot) \wedge (\lnot q \vee \top)] \} \wedge  \{\lnot p \wedge [ (r \vee ( (q \vee \bot) \wedge (\lnot q \vee \bot) )) \wedge (\lnot r \vee ((q \vee \top) \wedge (\lnot q \vee \bot) )) ] \}$$

\end{footnotesize}

\bigskip

Now, this is very hard on the human eye. It can be obviously simplified, noticing that e.g. $q \vee \bot$ can be replaced with $q$.\footnote{In fact, a version of the program contains a function doing exactly that. See the section describing the experiment for details.} Still, it is a far cry from the example as presented in the theoretical part.

\subsection{The program}

Our aim was to develop \emph{a proof-of-concept} implementation, rather than a fully-fledged one. This means that we eschew any optimisation attempts, including the simplifications described above. For that reason, we also decided to implement the procedure in the case of formulas with at most four variables. This has the benefit of providing a relatively light-weight interpolant finder. It also has its drawbacks, as it will become evident in the section describing experiments run using this software. 

We now present a general set-up fit for any number of variables. Consider the following pseudocode. This function takes two formulas $X$,$Y$ as an input, s.t. $X \to Y$ is a tautology, $X$ is in a CNF and $Y$ is in a DNF.\\

\begin{algorithmic}
    \IF{$X$ and $Y$ have variables in common}
        \STATE create $Y^*$
	\STATE \IF{The set of pairs of literals in $X, Y^* \to \bot$ is not empty}
			\STATE randomly choose one pair of literals (say $l$, $l^*$)
			\STATE create $G_1$
			\STATE create $G_2$
			\STATE create an interpolant of the required form, containing $l$ and $l^*$ together with the recursive calls using $G_1$ and $G_2$
		\ELSE 
			\STATE \IF{$\lnot X$ is a tautology} 
					\STATE interpolant is $\bot$
				\ELSE
					\STATE interpolant is $\top$
				\ENDIF				
		\ENDIF
    \ELSE
        \STATE \IF{$X$ is not empty}
			\STATE \IF{$Y$ is not empty}
					\STATE \IF{$\lnot X$ is a tautology}
							\STATE interpolant is $\bot$
						\ELSE
							\STATE interpolant is $\top$
						\ENDIF
				\ELSE
					\STATE interpolant is $\bot$
				\ENDIF
		\ELSE
			\STATE interpolant is $\top$
		\ENDIF
    \ENDIF
\end{algorithmic}




We believe that the presented example gives a clear idea how the procedure operates in a particular case and given the general description above, one is able to easily figure out how the procedure behaves in situations that go beyond the example we have presented. Therefore, we feel there is no need for a more involved description of all the possibilities, as it also bears the risk of unnecessarily muddling things up. Perhaps we should add a word on how literal elimination works. Initially we have \texttt{X} to the left and \texttt{Y} to the right of the implication. We then move  \texttt{Y} left, creating \texttt{Y'} (by changing disjuncts in CNF to conjuncts in DNF). It is in this setting where the literal elimination takes place. The remaining formulas are then redistributed to new \texttt{X}'s and \texttt{Y}'s (where the remaining conjuncts from \texttt{Y'} are converted back to disjuncts), in accordance with the rules specified for the creation of $G_1$ and $G_2$.\footnote{We remark that whereas the theoretical procedure starts with formulas where literals and their complements do not occur in one clause and where there is at most one occurrence of a given literal in a clause, our implementation is less picky when it comes to the input formulas. We need to ensure, however, that the rank decreases with each step. To prevent this, all occurrences of a given literal in a given clause are eliminated \emph{simultaneously}. Consider the following example:\\

\begin{scriptsize}


\begin{center}

\texttt{X:  [D.p..Np., D.p..p.] Y:  [C.p.]}\\
\bigskip
\end{center}

\texttt{X:  [D.Np., .0.] Y:  []} \hfill \texttt{X:  [D.p.] Y:  [.1.]}\\

\end{scriptsize}

Here, the fact that all occurrences of \texttt{.p.} are eliminated at once from \texttt{D.p..p.} meant that it changes into \texttt{.0.} in the left-hand formula. Note that the input formula is a non-starter in terms of theoretical considerations and needs to be pre-processed first. This does not have to happen in the presented implementation. The procedure returns $(p \vee \bot) \wedge (\lnot p \vee \top)$ as an interpolant. (The input formula is obviously equivalent to \texttt{X:  [D.p.] Y:  [C.p.]}, which is analysed in the body of the text.)}

A note on the correctness of the described algorithm. In the base cases, the theoretical result guarantees that either a negation of the left side is a tautology or the right side is so. Thus, we only have to check which of the two possibilities obtains.\footnote{By the same token, if one of the sides is empty, we do not have to check the validity of the remaining formula (or its negation) --- it is guaranteed by what is set out in the proof of the main theorem.} The next step ensures that the algorithm builds a proper interpolant on top of the base case results (as specified in the proof of the main theorem). Looking at the pseudocode, one might wonder what happens in the case, when both sides of the implication are empty. After all, this situation is not covered! We remind the reader that by the way $G_1$ and $G_2$ are constructed, as mentioned in the description of the example, empty conjuncts and --- in effect --- disjuncts are replaced by $\bot$ and $\top$, respectively. This ensures that we never get to the situation where both sides of the implication are empty. To see that, consider the following example. The only way, in which we could possibly obtain two empty lists is something akin to \texttt{X: [D.p.] Y:  [C.p.]}.\footnote{Note that e.g. \texttt{X: [D.p.] Y:  [C.Np.]} would mean that there are no pairs of literals to be eliminated from \texttt{X} and \texttt{Y'}.}
 Here, \texttt{Y'} would only contain \texttt{D.Np.}, which is a one-element DNF. The pair of literals to be eliminated is, obviously, \texttt{.p.} and \texttt{.Np.}. It might seem as if we should obtain the empty sets by doing so, but in fact it looks as follows.\footnote{The procedure returns the formula $(p \vee \bot) \wedge (\lnot p \vee \top)$ as the interpolant (the same as in the case of \texttt{X:  [D.p..Np., D.p..p.] Y:  [C.p.]}). This can be easily simplified to $p$, which agrees with our human intuitions.}\\

\begin{scriptsize}


\begin{center}

\texttt{X:  [D.p.] Y:  [C.p.]}\\
\bigskip
\end{center}

\texttt{X:  [.0.] Y:  []} \hfill \texttt{X:  [] Y:  [.1.]}\\

\end{scriptsize}

The constants \texttt{.0.} and \texttt{.1.} appear in the resulting formulas since we either eliminate the final one-element disjunct from \texttt{X} or \texttt{Y'} (which is \texttt{[D.Np.]}) (which then gets transferred to the other side of implication as \texttt{Y}), meaning that there can be only one empty set in each of the leaves of this tree (it should be clear how the empty sets come to be).

One final point concerns a comparison of HKP-system with the one we are presenting here. Despite superficial similarities --- the fact that the interpolants are constructed using the same schemas --- there is a difference that can potentially be of importance. As mentioned above, we do not make use of binary resolution when constructing interpolants and instead utilise a refutation method akin to non-binary resolution. It seems that owing to this, there is potential to produce interpolants in a fewer number of steps. Consider the following example mentioned in \cite{DSilva2010}. One is to find an interpolant for $(p \vee \lnot q) \wedge (\lnot p \vee \lnot r) \wedge q$ and $(\lnot q \vee r) \wedge (q \vee s) \wedge \lnot s$. The accompanying diagram in \cite{DSilva2010} shows that one requires five steps in producing interpolant using HKP-system. In our notation this example becomes \texttt{X:  [D.p..Nq.,D.Np..Nr.,D.q.]} and \texttt{Y:  [C.q..Nr.,C.Nq..Ns.,C.s.]}. The entire procedure can take as little as two steps: by eliminating first the pair \texttt{.q.}, \texttt{.Nq.} and then the pair \texttt{.r.}, \texttt{.Nr.} we conclude the interpolant search.

As already mentioned, our implementation is in a form of a simple Python script, designed to be run on essentially any reasonable Linux distribution.\footnote{The program we developed is available for download at the following location: \url{https://iphils.uj.edu.pl/~a.trybus/code/interpolant.zip}.} At its core, it has a function \texttt{interpolant(X,Y)} implementing the procedure described above in a relatively straightforward manner. Apart from that, the main file contains a number of helper functions, in most cases allowing one to parse or translate the formulas into different notation for further processing or functions used within the main function in shuffling parts of the formulas back and forth, selecting the desired elements, removing them etc. The main file also contains additional functions used in pretty-printing the results. It is accompanied by other files: used for tautology checking, collecting and plotting the experimental data, etc.

\subsection{Experimental results}

We run a small number of tests involving the described procedure. The experiment was divided into two parts. In the first part, the scope of those tests was limited by the constraint on the number of variables. We designed a pseudo-random formula generator with adjustable number of conjuncts ($c > 0$), disjuncts ($d > 0$). Each time when a conjunct/disjunct is generated, the program decides on the number of variables it is to contain (less than or equal to $4$).\footnote{Note that by design, the variables or their negations cannot occur more than once in a given disjunct.} with a $50\%$ chance for a given variable being negated. Thus we end up with two lists, representing $X$ and $Y$. The final constraint was that the formula $X \to Y$ is a tautology (only then the interpolant search is reasonable). We then run the program in the loop in two different settings: (a) producing $100000$ formulas with $c, d \leq 10$, and (b) producing $50000$ formulas with $c,d = 20$. Finally, all generated formulas were plugged in to \texttt{interpolant(X,Y)} and the results of the interpolant search were recorded.\footnote{The datasets are made available at the location mentioned above. Note that by design. The function \texttt{interpolant(X,Y)} produces interpolants in a pseudo-random fashion, in each run pseudo-randomly selecting a potential literal to be eliminated.} We tracked execution time (in seconds), size of $X,Y$ and of the resulting interpolants, the number of connectives (including negations) and number of variables in all formulas.\footnote{We run our tests on a machine with Intel Celeron CPU 1000M @ 1.80GHz and 4GB RAM.} For the case (a) the average execution time (rounded to two significant figures) was $0.0032$ and size of interpolant (rounded to the nearest natural number) was  $165$; in case (b) it was $0.011$ and $346$ respectively. Figure \ref{figex1} shows the relationship between the execution time and the size of $X \to Y$ in cases (a) and (b). We have removed a handful of outliers (most likely resulting from quirks in the measurement process) for case (a) to increase the visibility of the main trends in data. The unadulterated datasets are also made available at the described location. Figure \ref{figex2} shows the relationship between the size of interpolant and the execution time for cases (a) and (b). 

\begin{figure}[h!]
    \centering

        \subfigure[]{\includegraphics[width=0.45\textwidth]{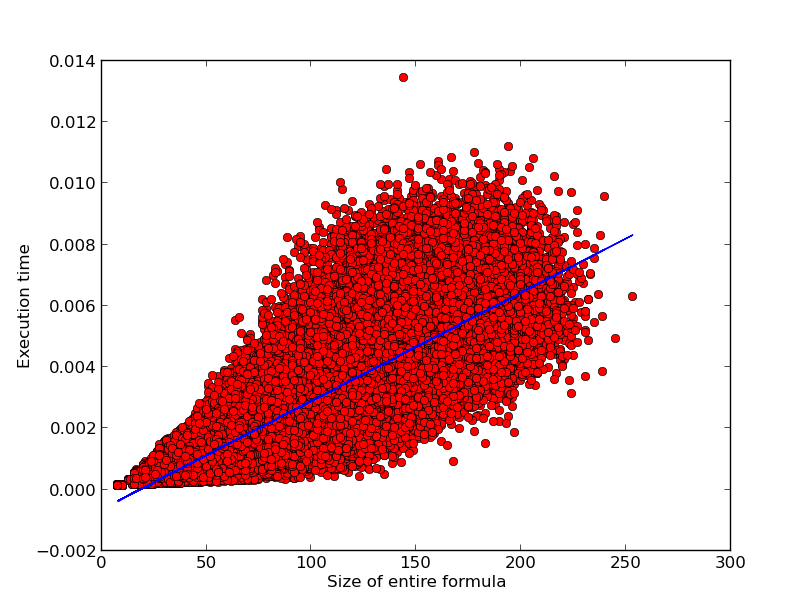}}
        \subfigure[]{\includegraphics[width=0.45\textwidth]{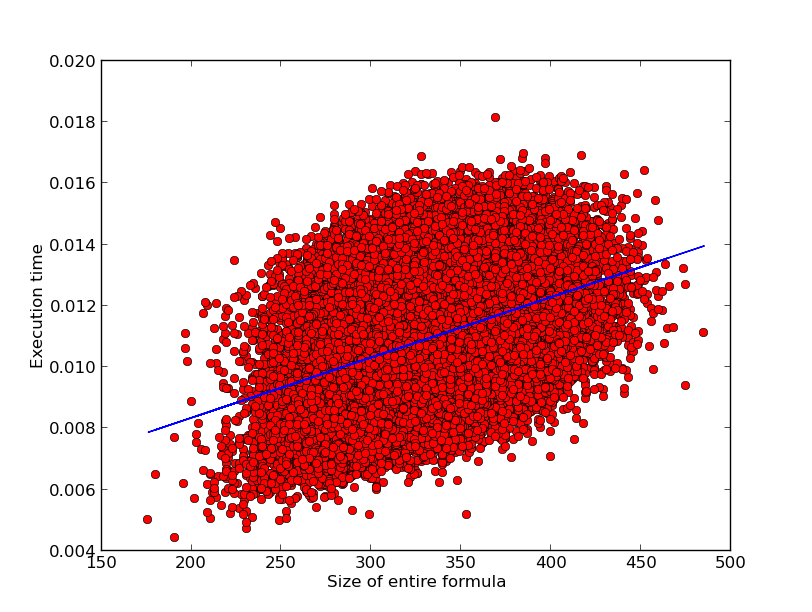}}

    \caption{Execution time [sec] ($y$ axis) and the size of $X \to Y$ ($x$ axis) for various experimental settings. Regression line in blue.}\label{figex1}
\end{figure}

\begin{figure}[h!]
    \centering

        \subfigure[]{\includegraphics[width=0.65\textwidth]{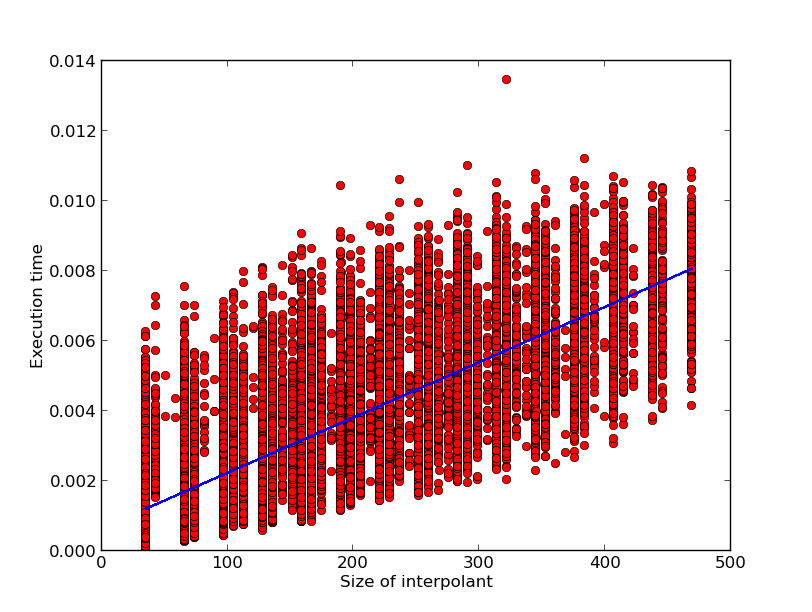}}
        \subfigure[]{\includegraphics[width=0.65\textwidth]{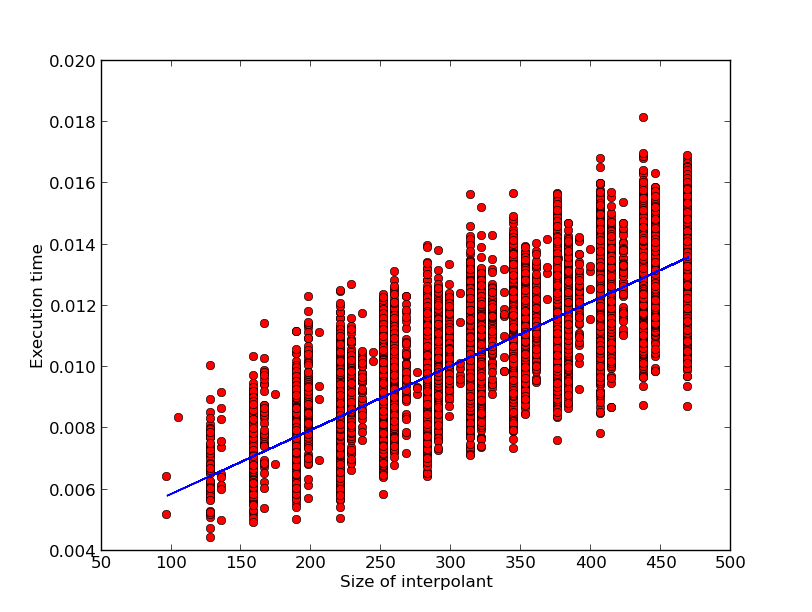}}\\

    \caption{Execution time [sec] ($y$ axis) and the size of interpolant ($x$ axis). Regression line shown in blue.}\label{figex2}
\end{figure}

Since in all the plots the regression line fits well in the general shape of data-points, this suggest a linear relationship between the values of $x$ and $y$. No doubt more tests are required to establish a more definite relationship. We feel that increasing the values of $c$ and $d$ might not result in radical changes in the tracked data, since with the increase in $c$ and $d$ values, the formulas created become more and more repetitive, due to the fact that we limited ourselves to four variables, potentially reducing the computational burden.

The second part of the experiment involved a smaller sample of formulas, but the restriction on the number of variables was lifted. The only changes in the implementation involved an extension allowing the use of essentially any number of variables represented as \texttt{P(a,b)}, where \texttt{a} is the number assigned to this specific variable and \texttt{b} is the total number of variables in the given formula (of course \texttt{a} $\leq$ \texttt{b}).\footnote{The extended code is also made available (\texttt{interpolant-extended.py}). As alluded to before, it also contains a function \texttt{simplify(X)} taking a formula (in effect, an interpolant) in the format accepted by our parser and simplifying it by reducing the conjuncts and disjuncts containing verum or falsum constants in the well-known way. Note that this procedure can easily be a built-in feature of the interpolant-finding function itself, thus simplifying the resulting interpolant on the fly.} So, for example \texttt{P(1,2)} represents the variable number one out of two variables total whereas \texttt{P(13,14)} is the variable number thirteen out of fourteen in total, etc. Note that the the first example can simply be rendered as \texttt{p} whereas the second has no similar simple representation. In this part of the experiment, the algorithm has been tested on $1000$ randomly-generated formulas with at most $10$ variables, $20$ conjuntcs and $20$ disjuncts. The results, presented in Fig. \ref{figfinal} seem to confirm the tendency noticed in the previous experimental setting with fewer variables, namely that of a linear performance of the implementation (the relation between execution time and size of interpolant) with some results indicating a more complicated situation (the relation between execution time and size of the original formula).

\begin{figure}[h!]
    \centering

        \subfigure[]{\includegraphics[scale=0.37]{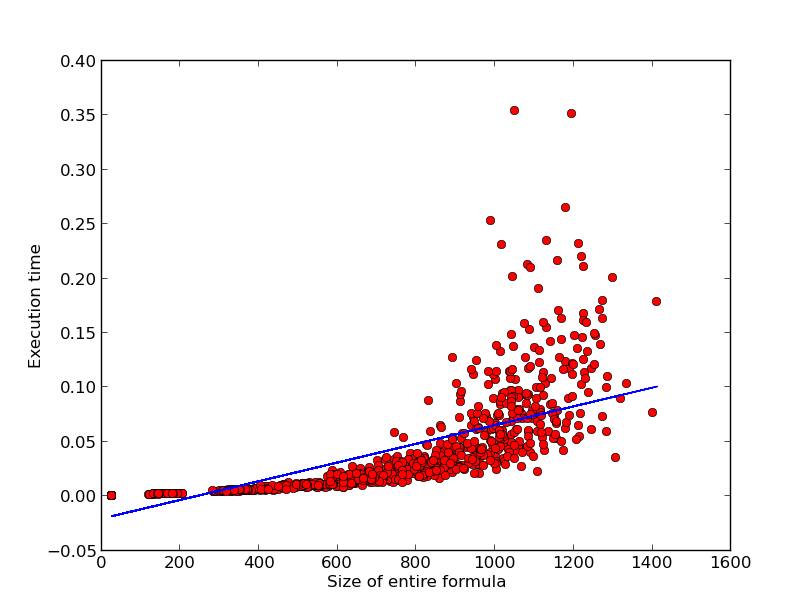}}
        \subfigure[]{\includegraphics[scale=0.37]{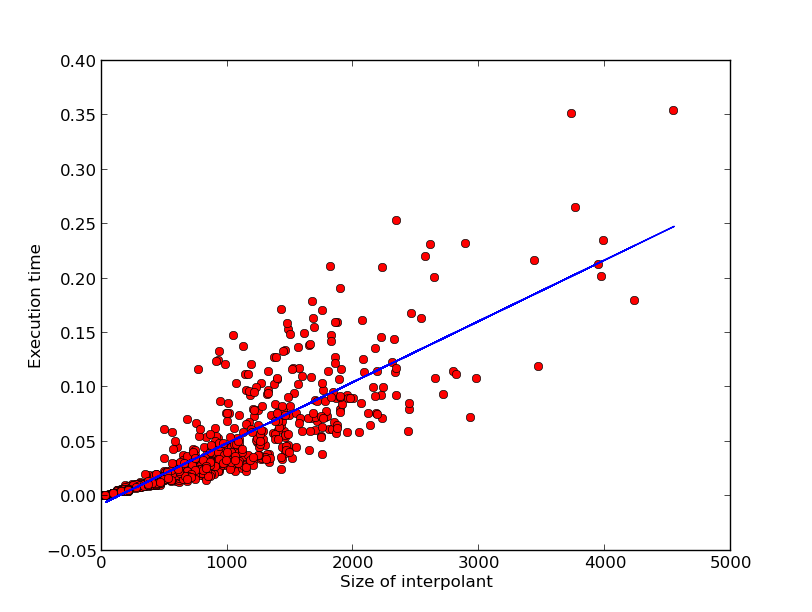}}

    \caption{Execution time [sec] and (a) the size of $X \to Y$ (b) the size of the interpolant. Regression lines in blue.}\label{figfinal}
\end{figure}

\section{Conclusions}

This article described presented a conceptually simple method for finding interpolants for formulas in classical propositional logic. In contrast to the standard approaches, our method comes from the research on the refutation systems. Since the proof of the interpolation theorem is constructive, it lends itself well to potential practical uses. We took advantage of that and proposed a proof-of-concept interpolant finding program. The interpolants produced by our program though corresponding well to the pure theoretical procedure are perhaps not quite the eye candy. Our aim, however, was simply to show the practical applicability of the theoretical results and therefore there is, obviously, still room for improvement and growth. Our system has some potential advantages over similar propositions: as a result of it not being limited by the binary resolution it can produce results in a fewer number of steps. However, since the approaches presented in the literature dealt with first-order interpolant systems, the most important challenge ahead would be to extend our system beyond propositional logic. This theoretical challenge is, obviously, coupled with the more practical one of designing a relevant computer program.

\bibliographystyle{unsrtnat}
\bibliography{references}  






\end{document}